\documentclass[11pt]{article}
\usepackage[utf8]{inputenc}
\usepackage[T1]{fontenc}
\usepackage{amsmath,amssymb}
\usepackage{indentfirst}
\usepackage{mathrsfs}
\usepackage{amsthm}
\usepackage{thmtools}
\usepackage{graphicx}
\usepackage{caption}
\usepackage[usenames, dvipsnames]{color}
\usepackage{csquotes}
\usepackage{bookmark}
\usepackage{booktabs}
\usepackage{url}
\usepackage{physics}
\usepackage{enumerate}
\usepackage{dsfont}
\usepackage[margin=1.0in]{geometry}
\usepackage{enumitem}
\setlist[itemize]{leftmargin=*}
\usepackage{mathtools}
\usepackage{newtxmath}
\usepackage{newtxtext}
\usepackage{xcolor}

\theoremstyle{plain}
\newtheorem{theorem}{Theorem}
\newtheorem{lemma}[theorem]{Lemma}
\newtheorem{proposition}[theorem]{Proposition}
\newtheorem{corollary}[theorem]{Corollary}

\newcommand{\propref}[1]{Proposition~\ref{#1}}
\newcommand{\lemref}[1]{Lemma~\ref{#1}}
\newcommand{\secref}[1]{\S~\ref{#1}}

\newcommand{\kl}{\text{QRE}}
\newcommand{\divg}[2]{\mathsf{D}(#1||#2)}
\newcommand{\divgbig}[2]{\mathsf{D}\big(#1\rVert#2\big)}
\newcommand{\divgchisq}[2]{\mathsf{D}_{{\kappa}}(#1\rVert#2)}
\newcommand{\divgkl}[2]{\mathsf{D}_{\kl}(#1\rVert#2)}
\newcommand{\divgklbig}[2]{\mathsf{D}_{\kl}\big(#1 \big\lVert #2\big)}
\newcommand{\divgarg}[3]{\mathsf{D}_{#1}(#2\rVert#3)}
\newcommand{\divgargbig}[3]{\mathsf{D}_{#1}\big(#2\rVert#3\big)}
\newcommand{\dvg}{\mathsf{D}}

\newcommand{\Real}{\mathbb{R}}
\newcommand{\unit}{\mathbb{I}}
\newcommand{\ce}{\mathcal{E}}
\newcommand{\cn}{\mathcal{N}}

\newcommand{\hbt}{\mathscr{H}}
\newcommand{\dm}{\mathscr{D}}
\newcommand{\vnentropy}{\mathsf{H}}
\newcommand{\vnentropymin}{\mathsf{H}^{\text{min}}}

\newcommand{\myeq}[1]{\stackrel{{#1}}{=}}
\newcommand{\mygt}[1]{\stackrel{{#1}}{>}}
\newcommand{\myle}[1]{\overset{\scriptstyle#1}{\leq}}

\newcommand{\eps}{\epsilon}

\usepackage{hyperref}
\hypersetup{
  colorlinks=true,
  linkcolor=blue!50!black!,
  urlcolor=red!50!black!,
  citecolor=green!50!black!,
}

\usepackage[numbers]{natbib}
\usepackage{hyperref}
\usepackage{doi}

\newcommand{\myaddress}{Institute of Natural Sciences, School of Mathematical Sciences, \& Ministry of Education Key Laboratory in Scientific and Engineering Computing, Shanghai Jiao Tong University, Shanghai 200240, China. Email \texttt{yucao@sjtu.edu.cn}}

\title{Counterexamples to the Tensorization Property of Strong Data Processing Inequalities for Quantum Divergences}
\author{Yu Cao \thanks{\myaddress}}

\date{}

\begin{document}

\maketitle

\begin{abstract}
In information theory, the data processing inequality is a fundamental property that describes the loss of information through noisy channels. A more refined description is given by the strong data processing inequality (SDPI) for a fixed channel and a full-rank reference state. In classical information theory, the strong data processing inequality tensorizes for a large family of $f$-divergences. However, its quantum counterpart remains largely unexplored. Tensorization of SDPI has previously been shown only for some special cases, and whether it holds more generally for quantum divergences still remains open. In this work, we report two negative results: the tensorization property fails for certain quantum chi-square divergences, and it also does not hold for the quantum relative entropy.
\end{abstract}

\section{Introduction}

The data processing inequality is a fundamental principle in information theory, characterizing the loss of information through noisy channels \cite{cover2006elements,wilde_quantum_2017}. Denote a classical or quantum divergence by $\divg{\cdot}{\cdot}$. Given an arbitrary channel $\ce$ and two arbitrary physical states $\rho$ and $\sigma$, the divergence $\divg{\cdot}{\cdot}$ is said to satisfy the \emph{data processing inequality} if 
\begin{align*}
\divgbig{\ce(\rho)}{\ce(\sigma)} \le \divgarg{}{\rho}{\sigma},\qquad \forall \ce, \rho, \sigma.
\end{align*}
For a fixed channel $\ce$ and a full-rank state $\sigma$, if there exists a coefficient $\eta<1$ such that 
\begin{align*}
\divgbig{\ce(\rho)}{\ce(\sigma)} \le \eta \divgarg{}{\rho}{\sigma}, \qquad \forall \rho,
\end{align*}
then it is called the \emph{strong data processing inequality} (SDPI). The optimal coefficient is called the \emph{strong data processing inequality constant}, which is defined as follows
\begin{align*}
\eta_{\dvg}(\ce, \sigma) := \sup_{\rho \neq \sigma} \frac{\divgbig{\ce(\rho)}{\ce(\sigma)}}{\divg{\rho}{\sigma}}.
\end{align*}
For this quantity to be well-defined, one implicitly assumes that the reference states $\sigma, \ce(\sigma)$ are both full-rank.
The data processing inequality ensures that $\eta_{\dvg}(\ce, \sigma)\le 1$. 
A closely related quantity is called the contraction coefficient, which is defined as $\eta_\dvg(\ce) := \sup_\sigma \eta_\dvg(\ce, \sigma)$ by taking the supremum over all full-rank $\sigma$ \cite{hiai_contraction_2016}.

The above SDPI constant enjoys a beautiful tensorization property for a large family of classical $f$-divergences (whose $f$-entropy is sub-additive and homogeneous). In particular, 
\begin{align}
\label{eqn::tensor}
\eta_{\dvg}(\ce_1\otimes \ce_2, \sigma_1\otimes \sigma_2) = \max\big\{ \eta_{\dvg}(\ce_1, \sigma_1), \eta_{\dvg}(\ce_2, \sigma_2)\big\}.
\end{align}
See Raginsky's seminal work \cite[Theorem III.9]{raginsky_strong_2016} for detailed proofs.
This tensorization property means that the global rate of information loss for tensor product channels is fully characterized by the local rates.
We remark that the dimensions of the channels $\ce_i$ are not necessarily the same for $i = 1, 2$, and we focus on the finite-dimensional case herein. The conclusion \eqref{eqn::tensor} can thus be easily generalized to a finite number of channels.

Even though tensorization of SDPI constants for classical divergences is well-established, extending the theory to the quantum regime remains an open problem.
The quantum version of the tensorization property was investigated in \cite{cao_tensorization_2019}, which proved its validity for a special quantum chi-square divergence (cf. Lemma~\ref{lem::chisq_half}), and for quantum-classical channels with a particular family of quantum chi-square divergences. 
The SDPI constants and contraction coefficients for quantum divergences have also been recently studied systematically; see e.g., \cite{hirche_contraction_2022,george_unified_2026}. In particular, the contraction coefficient $\eta_\dvg(\mathcal{E})$ was proved not to tensorize for quantum relative entropy \cite[Sect.~3]{hirche_contraction_2022}, but whether SDPI constants tensorize \eqref{eqn::tensor} was still an open question for a general quantum divergence. A different but related variant of SDPI, called complete strong data processing inequality (CSDPI), was formulated in \cite{gao_complete_2022} and shown to be stable under tensorization; we remark that CSDPI is a stronger notion and differs from the setup herein.

In this paper, we show that the tensorization fails in general, therefore answering the open questions raised in \cite{cao_tensorization_2019,hirche_contraction_2022} in the negative.
More specifically, we show that there exist quantum channels $\ce_1$, $\ce_2$, full-rank density matrices $\sigma_1$ and $\sigma_2$ such that 
\begin{align*}
\eta_{\dvg}(\ce_1\otimes \ce_2, \sigma_1\otimes \sigma_2) > \max\big\{ \eta_{\dvg}(\ce_1, \sigma_1), \eta_{\dvg}(\ce_2, \sigma_2)\big\}.
\end{align*}
In the following sections, we consider two families of quantum divergences: the family of quantum chi-square divergences in \secref{sec::chisq}, and quantum relative entropy in \secref{sec::kl}. These are perhaps among the most widely used quantum divergences in quantum information theory. Thus, the negative results in this work suggest that one should not expect the tensorization of SDPI in general in the quantum setting.
However, this does not imply that the tensorization property cannot hold for a special quantum chi-square divergence, or for some special channels, as already proved in \cite{cao_tensorization_2019}. The main conclusions are summarized in Table~\ref{table::main}.

\begin{table}
\centering
\renewcommand{\arraystretch}{1.25}
\begin{tabular}{ccc}
\toprule
 & Classical & Quantum  \\
 \toprule
chi-square &  Yes & Yes for $\kappa_{1/2}$ (known from \cite{cao_tensorization_2019}) \\
 &   & No in general (see Prop.~\ref{prop::chisq}) \\
\hline
relative entropy & Yes & No (see Prop.~\ref{prop::kl}) \\
 \bottomrule
\end{tabular}
\medskip
\caption{Tensorization property of the SDPI in classical and quantum settings}
\label{table::main}
\end{table}

\subsection*{Notations}
We denote the $N$-dimensional Hilbert space by $\hbt_N$, and the corresponding space of density matrices by $\dm_N$. We denote the family of quantum chi-square divergences by $\divgchisq{\cdot}{\cdot}$ (associated with a function $\kappa$) \eqref{eqn::chisq}, quantum relative entropy by $\divgkl{\cdot}{\cdot}$ \eqref{eqn::qre}, von Neumann entropy by $\vnentropy(\cdot)$ \eqref{eqn::Hmin}, and minimum output entropy by $\vnentropymin(\cdot)$ \eqref{eqn::Hmin}; the exact definitions are given where they are first used. We denote the maximally entangled state on $\hbt_N \otimes \hbt_N$ by $\ket{\Psi_{N}} = \frac{1}{\sqrt{N}} \sum_{k=1}^N \ket{k}\ket{k}$, and the maximally mixed density matrix by $\pi_N = \frac{1}{N} \sum_{k=1}^N \ketbra{k} \in \dm_N$. The identity matrix is denoted by $\unit$, and Pauli matrices are denoted by $X, Y, Z$ respectively. For any density matrix $\rho$, its smallest eigenvalue is denoted by $\lambda_{\min}(\rho)$.

\section{Counterexample for quantum chi-square divergence}
\label{sec::chisq}

The classical chi-square divergence is a widely studied object in classical information theory. In the quantum regime, 
due to the noncommutative nature of quantum objects, there is no unique definition of quantum chi-square divergence. A family of quantum chi-square divergences was introduced in  \cite{temme_divergence_2010}:
\begin{align}
\label{eqn::chisq}
\divgchisq{\rho}{\sigma} = \tr\big((\rho-\sigma) \Omega_{\sigma}^{\kappa}(\rho-\sigma)\big), \qquad \Omega_{\sigma}^{\kappa} := R_{\sigma}^{-1} \kappa( L_{\sigma} R_{\sigma}^{-1}),
\end{align}
where $\kappa: \Real^{+}\to \Real^{+}$ comes from the following family:
\begin{align*}
\mathcal{K} = \Big\{\kappa: \Real^{+}\to \Real^{+}:\ -\kappa \text{ is operator monotone},\  \kappa(1) = 1,\ x\kappa(x) = \kappa(x^{-1}) \Big\}.
\end{align*}
Here, $L_\sigma$ and $R_\sigma$ are the left and right multiplication operators associated with $\sigma$, respectively. If $\sigma$ has the eigenvalue decomposition $\sigma = \sum_{k} \lambda_k \ketbra{k}$, then
\begin{align}
\label{eqn::Omega_sigma}
\Omega_{\sigma}^{\kappa}(\cdot) = \sum_{j,\ell} \kappa(\lambda_j/\lambda_\ell) \frac{1}{\lambda_\ell} \ketbra{j} \cdot \ketbra{\ell}.
\end{align}
For convenience, we denote the corresponding SDPI constants by $\eta_\kappa(\ce,\sigma)$.

There are a few special examples. 
An important sub-family is called the mean-alpha family:
\begin{align*}
\kappa_{\alpha}(x) =  \frac{1}{2} \big(x^{-\alpha} + x^{\alpha-1}\big).
\end{align*}
Due to symmetry, we will consider $\alpha\in [0, 1/2]$ only. 
When $\kappa_{1/2}(x) = x^{-1/2}$,
\begin{align*}
\divgchisq{\rho}{\sigma} =  \tr\big((\rho-\sigma) \sigma^{-1/2} (\rho-\sigma) \sigma^{-1/2}\big).
\end{align*}
This is perhaps one of the simplest symmetric cases. This special case has been proved to satisfy the tensorization property.

\begin{lemma}[\cite{cao_tensorization_2019}]
\label{lem::chisq_half}
The tensorization property \eqref{eqn::tensor} holds for the quantum chi-square divergence with $\kappa(x) = \kappa_{1/2}(x)$.
\end{lemma}

A natural question is whether a similar result holds for every $\kappa\in \mathcal{K}$ and for general quantum channels, which was left open in \cite{cao_tensorization_2019}.
In the following, we provide a negative result for the general case, which indicates that the tensorization property of SDPI for the classical case cannot be generally extended to the quantum regime, even within the family of quantum chi-square divergences. The counterexample was found near an endpoint of the mean-alpha family. Since all reference states are full-rank,
\begin{align*}
\kappa_{0}(x) = \frac{1}{2} \big(1 + x^{-1}\big),
\end{align*}
is well defined as $\alpha\to 0$.

\begin{proposition}
\label{prop::chisq}
The tensorization property does not hold for all $\kappa\in \mathcal{K}$ and quantum channels. More specifically, for $\kappa_0$ there exist a channel $\ce$ and a full-rank  state $\sigma$ such that 
\begin{align*}
\eta_{\kappa_0} (\ce\otimes \ce, \sigma \otimes \sigma) > \eta_{\kappa_0} (\ce, \sigma).
\end{align*}
\end{proposition}

Since $\kappa_{\alpha}$ is continuous in $\alpha$, which also implies the continuity of $\eta_\kappa$, one immediately has:
\begin{corollary}
\label{corollary}
There exists a neighborhood $[0, \alpha^\star)$ such that the tensorization fails for every $\alpha\in [0, \alpha^\star)$.
\end{corollary}

We remark that $\kappa_0 \ge \kappa_{1/2}$ satisfies the second condition in \cite[Theorem 1]{cao_tensorization_2019}. This explains the necessity to restrict channels $\ce$ to a special class in \cite{cao_tensorization_2019}.
The remainder of this section is devoted to the proof of \propref{prop::chisq} and Corollary~\ref{corollary}.

\subsection{Counterexample}
The counterexample could be constructed for a qubit system. Since the chi-square divergence has a quadratic form, it is straightforward to show that for any channel $\ce$ and reference state $\sigma$, one has 
\begin{align*}
\eta_{\kappa}(\ce,\sigma) = \sup_{H = H^\dagger, \tr(H) = 0, H \neq 0} \frac{\tr\big( \ce(H) \Omega_{\ce(\sigma)}^{\kappa} \ce(H)\big)}{\tr\big(H \Omega_{\sigma}^{\kappa} (H)\big)}.
\end{align*}
Let us pick $\kappa = \kappa_0$, so that 
\begin{align}
\label{eqn::chisqk0}
\eta_{\kappa_0}(\ce,\sigma) = \sup_{H = H^\dagger, \tr(H) = 0, H \neq 0} \frac{\tr\big( \ce(H)^2 \ce(\sigma)^{-1}\big) }{\tr\big( (H)^2 (\sigma)^{-1}\big)}.
\end{align}

{\noindent \emph{Step 1: single-qubit case.}} Choose $\sigma = \sigma_s = \frac{1}{2} \big( \unit + s Z)$, $s\in (-1,1)$, and let $\ce = \ce_{\gamma}$ be the amplitude-damping channel with Kraus operators
\begin{align*}
\ce_\gamma(\cdot) = K_0 (\cdot) K_0^\dagger + K_1 (\cdot) K_1^\dagger, \qquad K_0 = \begin{bmatrix} 1 & 0 \\ 0 & \sqrt{1-\gamma}\end{bmatrix} \qquad K_1 = \begin{bmatrix} 0 & \sqrt{\gamma} \\ 0 & 0 \end{bmatrix}, \qquad \gamma\in (0,1).
\end{align*}
It can be readily verified that 
\begin{align*}
\ce_\gamma(\unit) =  \begin{bmatrix} 1+ \gamma & 0 \\ 0 & 1-\gamma \end{bmatrix}, \qquad 
\ce_\gamma(X) &=  \sqrt{1-\gamma} X, \qquad \ce_\gamma(Y) = \sqrt{1-\gamma} Y \qquad \ce_\gamma(Z) = (1-\gamma) Z,
\end{align*}
and subsequently, by parameterizing $H = r_x X + r_y Y + r_z Z$,
\begin{align*}
\ce_{\gamma}(\sigma_s) &=\  \sigma_{f(s)},\qquad f(s) = \gamma + s - \gamma s\\
\tr\big( H^2 \sigma_s^{-1}\big) &= \big(r_x^2 + r_y^2 + r_z^2\big) \tr(\sigma_s^{-1} ) =  \big(r_x^2 + r_y^2 + r_z^2\big) \frac{4}{1-s^2} \\
\tr\big( \ce_\gamma (H)^2 \ce_\gamma(\sigma_s)^{-1}\big) &= \big((1-\gamma) (r_x^2 + r_y^2) + (1-\gamma)^2 r_z^2  \big) \frac{4}{1-f(s)^2}.
\end{align*}
Hence, we have
\begin{align*}
\eta_{\kappa_0}(\ce_\gamma,\sigma_s) =&\ \sup_{(r_x, r_y, r_z)\ \neq\ 0}\ \frac{\big((1-\gamma) (r_x^2 + r_y^2) + (1-\gamma)^2 r_z^2  \big) \frac{4}{1-f(s)^2}}{\big(r_x^2 + r_y^2 + r_z^2\big) \frac{4}{1-s^2}} \\
=&\ \frac{1-s^2}{1-f(s)^2} \sup_{(r_x, r_y, r_z)\ \neq\ 0}\ \frac{(1-\gamma) (r_x^2 + r_y^2) + (1-\gamma)^2 r_z^2 }{r_x^2 + r_y^2 + r_z^2} \\
=&\ \frac{1-s^2}{1-f(s)^2} \max\big\{1-\gamma, (1-\gamma)^2 \big\} 
=\ \frac{1-s^2}{1-f(s)^2} (1-\gamma).
\end{align*}
The last equality holds because $\gamma\in (0,1)$. One can also easily validate that $\ce(\sigma)$ is also full-rank.

\bigskip 
{\noindent \emph{Step 2: bipartite case.}}
Now, let us estimate $\eta_{\kappa_0}(\ce_\gamma\otimes \ce_{\gamma}, \sigma_s\otimes \sigma_s)$ and show that 
\begin{align}
\label{eqn::break}
\eta_{\kappa_0}(\ce_\gamma\otimes \ce_{\gamma}, \sigma_s\otimes \sigma_s) > \eta_{\kappa_0}(\ce_\gamma, \sigma_s),
\end{align}
when $s \in \big( \frac{\sqrt{1+(1-\gamma)^2}-1}{1-\gamma}, 1\big),\ \gamma\in(0,1)$.
This provides a sufficient condition for tensorization to break down. 

Let us pick the Hermitian matrix in \eqref{eqn::chisqk0} for the bipartite case as
\begin{align*}
\widetilde{H} = X \otimes X - Y \otimes Y.
\end{align*}
With this choice, the SDPI constant for the bipartite system can be lower bounded by
\begin{align*}
\eta_{\kappa_0}(\ce_\gamma\otimes \ce_{\gamma}, \sigma_s\otimes \sigma_s) \ge &\ \frac{\tr\Big( \big(\ce_\gamma\otimes \ce_{\gamma}(\widetilde{H})\big)^2 \ce_\gamma(\sigma_s)^{-1} \otimes \ce_\gamma(\sigma_s)^{-1}\Big)}{\tr\Big( \big(\widetilde{H}\big)^2 (\sigma_s)^{-1} \otimes (\sigma_s)^{-1}\Big)} \\
=&\ \frac{(1-\gamma)^2 \tr\big(\widetilde{H}^2 \ce_\gamma(\sigma_s)^{-1} \otimes \ce_\gamma(\sigma_s)^{-1}\big)}{\tr\big(\widetilde{H}^2 (\sigma_s)^{-1} \otimes (\sigma_s)^{-1}\big)} \\
=&\ \frac{(1-\gamma)^2 \tr\big((\unit \otimes \unit + Z \otimes Z)\ \ce_\gamma(\sigma_s)^{-1} \otimes \ce_\gamma(\sigma_s)^{-1}\big)}{\tr\big((\unit \otimes \unit + Z \otimes Z)\ (\sigma_s)^{-1} \otimes (\sigma_s)^{-1}\big)}\\
=&\ (1-\gamma)^2 \frac{(1+f(s)^2) (1-s^2)^2}{(1+s^2)(1-f(s)^2)^2}.
\end{align*}
To ensure \eqref{eqn::break}, we just need
\begin{align*}
(1-\gamma)^2 \frac{(1+f(s)^2) (1-s^2)^2}{(1+s^2)(1-f(s)^2)^2} > \frac{1-s^2}{1-f(s)^2} (1-\gamma),
\end{align*}
namely,
\begin{align*}
(1-\gamma) \frac{(1+f(s)^2) (1-s^2)}{(1+s^2)(1-f(s)^2)} > 1.
\end{align*}
This holds when
\begin{align*}
(s-1)^2 \big(1-2 s + s^2 (\gamma-1) - \gamma) (-1+\gamma) \gamma > 0.
\end{align*}
When $s\in (-1,1)$ (so that $\sigma_s$ is full-rank), and the damping coefficient $\gamma\in (0,1)$, the above condition holds if $1-2 s + s^2 (\gamma-1) - \gamma <0$, namely, 
\begin{align*}
s \in \Big( \frac{\sqrt{1+(1-\gamma)^2}-1}{1-\gamma}, 1\Big).
\end{align*}
If we pick $s = 1/2$, the counterexample is still valid for any $\gamma \in (0,1)$.

\subsection{Proof of Corollary~\ref{corollary}}

Since the counterexample for $\kappa_0$ holds at a full-rank reference state $\sigma$ (whose smallest eigenvalue is denoted by $\lambda_{\min}(\sigma)$), a small perturbation of $\alpha$ does not significantly change the SDPI constants, leading to counterexamples for $\kappa_{\alpha}$ when $\alpha \approx 0$. Recall a known fact that $\kappa_\alpha(x)$ is monotonically decreasing in $\alpha$ for $\alpha \in [0, 1/2]$, so $\kappa_\alpha(x) \le \kappa_0(x)$ for any $x$.

Let us denote $\delta = \min\big\{\lambda_{\min}(\sigma), \lambda_{\min}\big(\ce(\sigma)\big)\big\}$. Suppose that we pick $\alpha^\star$ such that 
\begin{align}
\label{eqn::alpha1}
\kappa_{\alpha}(z)/\kappa_0(z) \in [1-\eps, 1], \qquad \forall z \in [\delta, 1/\delta],\qquad \forall \alpha \in [0, \alpha^\star],
\end{align}
for some $\eps \in (0,1)$, then by \eqref{eqn::Omega_sigma}, 
\begin{align*}
& \divgarg{\kappa_\alpha}{\rho}{\sigma} = \tr\big((\rho-\sigma) \Omega_{\sigma}^{\kappa_\alpha}(\rho-\sigma)\big) \ge (1-\eps) \divgarg{\kappa_0}{\rho}{\sigma}, \\
& \divgarg{\kappa_\alpha}{\ce(\rho)}{\ce(\sigma)} = \tr\big(\ce(\rho-\sigma) \Omega_{\ce(\sigma)}^{\kappa_\alpha} \ce(\rho-\sigma)\big) \le \divgarg{\kappa_0}{\ce(\rho)}{\ce(\sigma)},
\end{align*}
which yields the estimate 
\begin{align*}
\eta_{\kappa_\alpha}(\ce, \sigma) \le \frac{1}{1-\eps} \eta_{\kappa_0}(\ce, \sigma).
\end{align*}
Similarly, if we pick $\alpha^\star$ such that 
\begin{align}
\label{eqn::alpha2}
\kappa_{\alpha}(z)/\kappa_0(z) \in [1-\eps, 1], \qquad \forall z \in [\delta^2, 1/\delta^2], \qquad \forall \alpha \in [0, \alpha^\star],
\end{align}
then we have
\begin{align*}
\eta_{\kappa_\alpha}(\ce\otimes \ce, \sigma\otimes \sigma) \ge (1-\eps) \eta_{\kappa_0}(\ce\otimes \ce, \sigma\otimes \sigma).
\end{align*}
Note that the condition \eqref{eqn::alpha2} is stronger than \eqref{eqn::alpha1} so that we only need to worry about \eqref{eqn::alpha2}.

Combining these estimates, we have 
\begin{align*}
\eta_{\kappa_\alpha}(\ce\otimes \ce, \sigma\otimes \sigma) \ge \big(1-\eps\big)^2\  \frac{\eta_{\kappa_0}(\ce\otimes \ce, \sigma\otimes \sigma)}{\eta_{\kappa_0}(\ce, \sigma)} \eta_{\kappa_{\alpha}}(\ce, \sigma).
\end{align*}
To obtain a counterexample, we simply need to pick $\eps$ small enough that 
\begin{align*}
\big(1-\eps\big)^2 (1-\gamma) \frac{(1+f(s)^2) (1-s^2)}{(1+s^2)(1-f(s)^2)} > 1.
\end{align*}

For simplicity, let us pick $s=1/2$ and $\gamma = 9/10$ as a concrete example, so that $f(s) = 19/20$. 
The above inequality holds for $\eps = 3/40$. We now return to the choice of $\alpha^\star$. Since $\delta = 1/40$ in this case, it suffices to pick $\alpha^\star$ such that 
\begin{align*}
\kappa_\alpha(z)/ \kappa_0(z) = \frac{z^{1-\alpha}+z^\alpha}{z+1} \in [1-\eps, 1], \qquad \forall z\in [\delta^{2}, \delta^{-2}],\qquad \alpha\in [0, \alpha^\star]. 
\end{align*}
This is ensured as long as 
\begin{align*}
z^\alpha, z^{-\alpha} \ge 1-\eps, \qquad \forall z\in [\delta^2, \delta^{-2}],\qquad \alpha\in [0, \alpha^\star]. 
\end{align*}
By direct simplification, we just need $\alpha^\star = \frac{\ln(1-\eps)}{\ln(\delta^2)} \approx 0.01$. We remark that this only provides a concrete admissible $\alpha^\star$; characterizing the largest such $\alpha^\star$ requires further study.

\section{Counterexample for quantum relative entropy}
\label{sec::kl}

Quantum relative entropy is also a fundamental quantity to measure the discrepancy between two quantum states in quantum information theory \cite{umegaki_conditional_1962,Petz_quantum_2008,wilde_quantum_2017}. For two quantum states $\rho \ll \sigma$, we define it by 
\begin{align}
\label{eqn::qre}
D_{\kl}(\rho || \sigma) = \tr\big(\rho \ln \rho \big) - \tr\big(\rho \ln \sigma\big).
\end{align}
We denote the corresponding SDPI constant by $\eta_{\kl}(\ce,\sigma)$ for the quantum channel $\ce$ and the reference state $\sigma$.
As already mentioned above, we cannot expect the tensorization property to hold in general: 
\begin{proposition}
\label{prop::kl}
The tensorization property does not hold for quantum relative entropy and general quantum channels. More specifically, there exist two channels $\cn, \overline{\cn}$ mapping states from the Hilbert space $\hbt_{N+1}$ to $\hbt_N$  and a full-rank quantum state $\widehat{\sigma}_\eps \in \dm_{N+1}$ such that 
\begin{align*}
\eta_{\kl} (\cn\otimes \overline{\cn}, \widehat{\sigma}_\eps \otimes \widehat{\sigma}_\eps) > \eta_{\kl} (\cn, \widehat{\sigma}_\eps) \equiv \eta_{\kl}(\overline{\cn}, \widehat{\sigma}_\eps).
\end{align*}
\end{proposition}
The meanings of the bar and hat notation will become clear below.
The key ingredients of the construction are Hastings' construction of superadditivity \cite{hastings_superadditivity_2009}, and a technique connecting the entropy to SDPI constants. The remainder of this section is devoted to the proof.

\subsection{Hastings' construction}

One main ingredient is an existence conclusion from \cite{hastings_superadditivity_2009} about the nonadditivity of the minimum output entropy $H^{\text{min}}(\mathcal{E})$, which is defined as 
\begin{align}
\label{eqn::Hmin}
\vnentropymin(\mathcal{E}) := \min_{\ket{\psi}} \vnentropy\Big(\mathcal{E}\big(\ketbra{\psi}\big)\Big) \equiv \min_{\rho\in \dm_N} \vnentropy\big(\mathcal{E}(\rho)\big), \qquad \vnentropy(\rho) := -\tr\big(\rho \ln \rho\big).
\end{align}
The equality above follows immediately from the concavity of von Neumann entropy.

Hastings proved in \cite{hastings_superadditivity_2009} that there exist quantum channels $\ce$ and $\overline{\ce}$ of the following form
\begin{align}
\ce(\rho) &= \sum_{i=1}^D p_i U_i \rho U_i^\dagger, \qquad 
\overline{\ce}(\rho) = \sum_{i=1}^D p_i \overline{U}_i \rho \overline{U}_i^\dagger,
\end{align}
where $\big\{U_i\big\}_{i=1}^D$ are unitary matrices, $\big\{p_i\big\}_{i=1}^D$ is a discrete probability distribution, and $\overline{U}$ denotes the complex conjugate of $U$, such that the minimum output entropy is nonadditive, 
\begin{align}
\label{eqn::subadditive}
\vnentropymin(\ce\otimes \overline{\ce}) < \vnentropymin(\ce) + \vnentropymin(\overline{\ce}) = 2 \vnentropymin(\ce).
\end{align}
In particular, Hastings’ construction showed that 
\begin{align}
\label{eqn::hastings}
\vnentropy\big(\ce\otimes \overline{\ce}(\ketbra{\Psi_{N}})\big) <  2 \vnentropymin(\ce),
\end{align}
where $\Psi_{N}$ is the maximally entangled state on the bipartite system $\hbt_N\otimes \hbt_N$; see \cite[p.~256]{hastings_superadditivity_2009}.

By the equivalence theorem in \cite{shor_equivalence_2004}, this nonadditivity immediately suggests the failure of three other types of additivity or strong superadditivity. 
This counterexample is, in fact, connected to a more general problem, which asks whether additivity holds when Rényi entropies with order $p\in [0,\infty)$ replace the von Neumann entropy in \eqref{eqn::subadditive}; when the order $p>1$,  this problem is also known as the {maximal $p$-norm multiplicativity conjecture}. Werner and Holevo provided a counterexample for order $p > 4.79$ \cite{werner_counterexample_2002}, and later counterexamples were constructed for $p>1$ in \cite{hayden_counterexamples_2008} and $p \approx 0$ in \cite{cubitt_counterexamples_2008}; the above Hastings' construction resolves the case $p = 1$ \cite{hastings_superadditivity_2009}; a further generalization to all $p\ge 0$ was proved recently in \cite{leung_counterexamples_2026}. 

\begin{lemma}
\label{lem::hastings_coro}
The channel constructed above satisfies 
\begin{align}
\label{eqn::hastings_counter}
\divgklbig{\ce\otimes \overline{\ce} (\ketbra{\Psi_N})}{\pi_N \otimes \pi_N} > 2 \sup_{\rho\in \dm_N} \divgklbig{\ce(\rho)}{\pi_N}.
\end{align}
\end{lemma}

\begin{proof}
By definition,
\begin{align*}
\divgklbig{ \ce\otimes \overline{\ce} & (\ketbra{\Psi_N})}{\pi_N \otimes \pi_N} =\ \ln (N^2) - \vnentropy\big(\ce\otimes \overline{\ce} (\ketbra{\Psi_N})\big) \\
\mygt{\eqref{eqn::hastings}} &\ 2\ln (N) - 2 \vnentropymin(\mathcal{E}) 
\myeq{\eqref{eqn::Hmin}}\ \sup_{\rho\in \dm_N} 2\ln (N) - 2 \vnentropy\big(\mathcal{E}(\rho)\big) \\
&\qquad =\ 2 \sup_{\rho\in \dm_N} \divgklbig{\ce(\rho)}{\pi_N}. \qedhere
\end{align*}
\end{proof}

For later convenience, let us denote
\begin{align}
\beta &:= \frac{\divgklbig{\ce\otimes \overline{\ce} (\ketbra{\Psi_N})}{\pi_N \otimes \pi_N}}{2 \sup_{\rho\in \dm_N} \divgklbig{\ce(\rho)}{\pi_N}} > 1, \label{eqn::beta}\\
C &:= \sup_{\rho\in \dm_N} \divgklbig{\ce(\rho)}{\pi_N}\label{eqn::C}.
\end{align}

It can also be easily verified that:
\begin{lemma}
The above channels $\ce$ and $\overline{\ce}$ are unital: 
\begin{align}
\label{eqn::cepi}
\ce (\pi_N) = \overline{\ce}(\pi_N) = \pi_N.
\end{align}
\end{lemma}
The existence result for the above $\ce$ and $\overline{\ce}$ satisfying the two properties \eqref{eqn::hastings_counter} and \eqref{eqn::cepi} is sufficient for our purposes below; in other words, the counterexample below only relies on a complex-conjugate pair of unital quantum channels that violate the additivity of minimum output entropy.
More explicit estimate of the dimension needed to observe a violation of the additivity have been studied in e.g., \cite{fukuda_comments_2010,belinschi_almost_2016}. Other possible probabilistic constructions of the above channel can be formulated, for instance, through dilation to a closed system \cite{brandao_hastings_2010}.

\subsection{Counterexample}

The counterexample is constructed for a channel mapping states from the Hilbert space $\hbt_{N+1}$ to $\hbt_N$, based on the constructed channel $\ce$ above; the key challenge is to connect the SDPI constants to  the minimum output entropy in a certain limit. 

To explain how to overcome this design challenge, we first briefly explain the idea behind constructing the counterexample below. 
Given the Hastings' construction, we already have \eqref{eqn::hastings_counter}, which resembles the numerators in the SDPI for quantum relative entropy. 
Therefore, we seek to construct a channel $\cn$ such that 
$\divgkl{\cn\otimes \overline{\cn}(\widehat{\rho}^{AB})}{\cn\otimes \overline{\cn}(\widehat{\sigma}\otimes \widehat{\sigma})}$
is connected to $\divgkl{\ce\otimes \overline{\ce}(\ketbra{\Psi_N})}{\pi_N\otimes \pi_N}$ in \eqref{eqn::hastings_counter} for some bipartite $\widehat{\rho}^{AB}$ and some full-rank state $\widehat{\sigma}$. 
The notation $\widehat{\sigma}$ emphasizes that $\widehat{\sigma}$ lies in a Hilbert space with different dimension, thereby enabling the setup to be more general; similar notation conventions apply to other symbols.
It is then natural to require $\cn(\widehat{\sigma}) = \overline{\cn}(\widehat{\sigma}) = \pi_N$. 
In addition, we need $\cn\otimes \overline{\cn}(\widehat{\rho}^{AB}) = \ce \otimes \overline{\ce} (\ketbra{\Psi_N})$ for some bipartite state $\widehat{\rho}^{AB}$. To achieve this, a natural idea is to define $\widehat{\rho}^{AB}$ as a block matrix containing $\ketbra{\Psi_N}$, so that the channel can map this block using the channel $\ce$. It turns out that one extra input dimension is enough.

More specifically, for any density matrix on $\hbt_{N+1}$, one can decompose it in terms of 
\begin{align*}
\widehat{\rho} = \begin{bmatrix} \varrho & v \\ v^\dagger & \rho \end{bmatrix} = \varrho \ketbra{0} + \sum_{k=1}^{N} \big(v_{k} \ketbra{0}{k} + h.c.\big) + \sum_{j, k = 1}^N  \ketbra{j} \rho\ketbra{k} \in \dm_{N+1},
\end{align*}
where $\varrho \in [0,1]$, $\rho$ is positive semi-definite, and $\text{h.c.}$ means Hermitian conjugate of the previous term. 
To simplify notation, we use $\ket{0}$ as the basis vector for the extra dimension, and $\ket{j}$ ($j = 1, 2, \cdots, N$) as the basis vectors for $\hbt_N$. In the following, density matrices in $\dm_{N+1}$ are denoted by hatted symbols, whereas their components on Hilbert space $\hbt_N$ are denoted by regular math symbols for clarity.
Then one defines 
\begin{align}
\label{eqn::cn}
\cn(\widehat{\rho}) := \varrho \pi_N + \ce(\rho).
\end{align}
The second term in this definition is motivated by the above discussion connecting $\cn$ to $\ce$; without the first term $\varrho \pi_N$, the above term would not be a channel, since $\rho$ is unnormalized and is not itself a density matrix.
Recall that we want $\cn(\widehat{\sigma}) = \pi_N$. Since $\ce$ is unital and maps $\pi_N$ to $\pi_N$, it is not surprising to choose the reference state as
\begin{align*}
\widehat{\sigma}_\epsilon := \begin{bmatrix} 1-\epsilon & 0 \\ 0 & \epsilon \pi_N\end{bmatrix} = (1-\eps) \ketbra{0} + \frac{\eps}{N}\sum_{k=1}^N \ketbra{k}.
\end{align*}
A similar construction using such a perturbed distribution was studied in, e.g., \cite{caputo_entropy_2025,hirche_contraction_2022}.

To verify that this mapping $\cn$ is a quantum channel, one can rewrite it in the following Kraus form:
\begin{align*}
\cn\big(\widehat{\rho}\big) = \sum_{j=1}^N \frac{1}{N} \ket{j} \bra{0} \widehat{\rho} \ket{0} \bra{j} + \sum_{j=1}^{D} p_j U_j \big(P \widehat{\rho} P^\dagger\big) U_j^\dagger,
\end{align*}
where $P$ is a mapping from $\hbt_{N+1} = \text{span}\big\{\ket{0}, \ket{1}, \cdots, \ket{N}\big\}$ to $\hbt_N = \text{span}\big\{\ket{1}, \cdots, \ket{N}\big\}$ as $P := \sum_{j=1}^{N} \ketbra{j}$.

Similarly, the conjugate channel is defined by
\begin{align*}
\overline{\cn}(\widehat{\rho}) := \varrho \pi_N + \overline{\ce}(\rho),
\end{align*}
which is also a quantum channel. 
It can be easily verified that 
\begin{align}
\label{eqn::N_sigma}
\cn(\widehat{\sigma}_\epsilon) = \overline{\cn}(\widehat{\sigma}_\epsilon) \myeq{\eqref{eqn::cepi}} \pi_N.
\end{align}

Next we choose the maximally entangled state $\widehat{\rho}_{AB} = \ketbra{\Psi_N} = \frac{1}{N} \sum_{j, k=1}^N \ket{j j} \bra{kk}$, and thus 
\begin{align*}
\eta_{\kl} (\cn\otimes \overline{\cn}, \widehat{\sigma}_\eps \otimes \widehat{\sigma}_\eps) \ge&\  \frac{\divgklbig{\cn\otimes \overline{\cn}(\widehat{\rho}_{AB})}{\cn\otimes \overline{\cn}(\widehat{\sigma}_\eps\otimes \widehat{\sigma}_\eps)}}{\divgklbig{\widehat{\rho}_{AB}}{\widehat{\sigma}_\eps\otimes \widehat{\sigma}_\eps}} \\
=&\ \frac{\divgklbig{\ce\otimes \overline{\ce}(\ketbra{\Psi_N})}{\pi_N\otimes \pi_N}}{2\ln(N) + 2 \ln (1/\eps)} \\
\myeq{\eqref{eqn::beta},\eqref{eqn::C}}&\ \beta \frac{C}{\ln(N) + \ln (1/\eps)}.
\end{align*}

Next we need to estimate the SDPI constants for each sub-system.

\begin{lemma}
\label{lem::eta_kl}
For any $\eps\in (0,1)$, one has $\eta_{\kl}(\cn, \widehat{\sigma}_\eps) = \eta_{\kl}(\overline{\cn}, \widehat{\sigma}_\eps)$, and in the limit of small $\eps$, one has the following:
\begin{align*}
\limsup_{\eps\to 0} \frac{\eta_{\kl}(\cn, \widehat{\sigma}_\eps)}{C/ \ln(1/\eps)} \le 1.
\end{align*}
\end{lemma}

Given this lemma (proved below), we know that, for sufficiently small $\eps$, 
$$\frac{\eta_{\kl}(\cn, \widehat{\sigma}_\eps)}{C/ \ln(1/\eps)} \le \frac{1+\beta}{2} \in (1, \beta),$$
so that 
\begin{align*}
\eta_{\kl} (\cn\otimes \overline{\cn}, \widehat{\sigma}_\eps \otimes \widehat{\sigma}_\eps) \ge &\ \frac{\beta C}{\ln(N) + \ln(1/\eps)} \\
= &\  \frac{\ln(1/\eps)}{\ln(N) + \ln(1/\eps)} \cdot \frac{\beta}{(1+\beta)/2} \cdot \frac{\frac{1+\beta}{2} C}{\ln(1/\eps)} \\
\ge &\ \frac{\ln(1/\eps)}{\ln(N) + \ln(1/\eps)} \cdot \frac{\beta}{(1+\beta)/2} \cdot \eta_{\kl}(\cn, \widehat{\sigma}_\eps)\\
> &\ \eta_{\kl}(\cn, \widehat{\sigma}_\eps),
\end{align*}
for sufficiently small $\eps$.
It remains to prove \lemref{lem::eta_kl}.

\subsection{Proof of Lemma~\ref{lem::eta_kl}}

\subsubsection*{Proof of invariance under complex conjugation}

If $\rho\in \dm_N$, then $\overline{\rho}$ is also a density matrix, and it is easy to validate that
$
\overline{\ce}(\overline{\rho}) = \overline{\ce(\rho)}$, for any $\rho\in\dm_N$.
Similarly, it is easy to verify that 
$\overline{\cn}(\overline{\widehat{\rho}}) = \overline{\cn(\widehat{\rho})}$, for any $\widehat{\rho}\in \dm_{N+1}$.
Consequently, it is straightforward to prove $\eta_{\kl}(\overline{\cn}, \widehat{\sigma}_\eps) = \eta_{\kl}({\cn}, \widehat{\sigma}_\eps)$, similar to \cite{hastings_superadditivity_2009}.

\subsubsection*{Proof of upper bound}

It was known from above that 
\begin{align*}
\eta_{\kl}({\cn}, \widehat{\sigma}_\eps) &= \sup_{\substack{\widehat{\rho}\in \dm_{N+1},\\ \widehat{\rho} \neq \widehat{\sigma}_\eps}} \frac{\divgklbig{\cn(\widehat{\rho})}{\cn(\widehat{\sigma}_\eps)}}{\divgkl{\widehat{\rho}}{\widehat{\sigma}_\eps}} \myeq{\eqref{eqn::N_sigma}}  \sup_{\substack{\widehat{\rho}\in \dm_{N+1},\\ \widehat{\rho} \neq \widehat{\sigma}_\eps}} \frac{\divgklbig{\cn(\widehat{\rho})}{\pi_N}}{\divgkl{\widehat{\rho}}{\widehat{\sigma}_\eps}}.
\end{align*}
Without loss of generality, write $\widehat{\rho}$ as:
\begin{align*}
\widehat{\rho} = \begin{bmatrix} 1-x & v \\ v^\dagger & x \rho\end{bmatrix}, \qquad \rho \in \dm_N,
\end{align*}
and define the dephasing channel $\mathcal{P}$ as follows,
\begin{align*}
\mathcal{P}(\widehat{\rho}) := \begin{bmatrix} 1-x & 0 \\ 0 & x \rho\end{bmatrix}.
\end{align*}
We can easily check that if $\mathcal{P}(\widehat{\rho}) = \widehat{\sigma}_\eps$, then $\cn(\widehat{\rho}) = \pi_N$, and hence $\divgklbig{\cn(\widehat{\rho})}{\pi_N} = 0$. This case does not contribute anything to computing SDPI constants, so that 
\begin{align*}
\eta_{\kl}({\cn}, \widehat{\sigma}_\eps) =  \sup_{\substack{\widehat{\rho}\in \dm_{N+1},\\ \mathcal{P} \widehat{\rho} \neq \widehat{\sigma}_\eps}} \frac{\divgklbig{\cn(\widehat{\rho})}{\pi_N}}{\divgkl{\widehat{\rho}}{\widehat{\sigma}_\eps}}.
\end{align*}

By plugging the expression of $\cn$ \eqref{eqn::cn} and the data processing inequality for quantum relative entropy to $\mathcal{P}$, one has 
\begin{align*}
\eta_{\kl}({\cn}, \widehat{\sigma}_\eps) &=   \sup_{\substack{\widehat{\rho}\in \dm_{N+1},\\ \mathcal{P}\widehat{\rho} \neq \widehat{\sigma}_\eps}} \frac{\divgklbig{(1-x) \pi_N + x \ce(\rho)}{\pi_N}}{\divgkl{\widehat{\rho}}{\widehat{\sigma}_\eps}}  \\
&\le  \sup_{\substack{\widehat{\rho}\in \dm_{N+1},\\ \mathcal{P}\widehat{\rho} \neq \widehat{\sigma}_\eps}}  \frac{\divgklbig{(1-x) \pi_N + x\ce(\rho)}{\pi_N}}{\divgkl{\mathcal{P}\widehat{\rho}}{\mathcal{P}\widehat{\sigma}_\eps}} \qquad \text{(data processing inequality)}\\
&= \sup_{\substack{x\in [0,1],\ \rho\in\dm_N,\\ (x, \rho) \neq (\eps, \pi_N)}}\ \frac{\divgklbig{(1-x) \pi_N + x \ce(\rho)}{\pi_N}}{d(x||\eps) + x \divgkl{\rho}{\pi_N}},
\end{align*}
where $d(x||\eps) := x\ln (x/\eps) + (1-x) \ln \big((1-x)/(1-\eps)\big)$. In this way, the supremum can be split into optimizing over $x$ and $\rho$ separately. Depending on the value of $x$, one should use different estimates: (a) when $x$ is relatively small, a perturbative expansion of $\divgklbig{(1-x) \pi_N + x \ce(\rho)}{\pi_N}$ has leading order $\order{x^2}$; (b) when $x$ is large, the key is to lower bound $d(x||\eps)$. \\

\emph{Case 1:} consider $x\in [0, {x}^{\star}]$, where $x^\star := \ln(1/\eps)^{-2}$, and $(x, \rho) \neq (\eps, \pi_N)$.
Then 
\begin{align*}
 &\ \divgklbig{(1-x) \pi_N + x \ce(\rho)}{\pi_N} 
=\ \divgklbig{\pi_N + x (\ce(\rho)-\pi_N)}{\pi_N} \\
\le&\ \divgargbig{2}{\pi_N + x (\ce(\rho)-\pi_N)}{\pi_N} 
=\ \ln\ \tr( \pi_N^{-1} \big(\pi_N + x (\ce(\rho) - \pi_N)\big)^2) \\
=&\ \ln \Big( \tr( \pi_N + x^2 (\ce(\rho) - \pi_N)^2 \pi_N^{-1} )\Big) 
=\ \ln \Big( 1 + x^2 \tr((\ce(\rho) - \pi_N)^2 \pi_N^{-1} )\Big) \\
\le&\ x^2 N \cdot \tr((\ce(\rho) - \pi_N)^2) 
\le\ x^2 N \cdot \big(\tr \abs{\ce(\rho) - \pi_N}\big)^2 \\
\le&\ 2\ln(2)\ x^2 N\ \divgkl{\ce(\rho)}{\pi_N},
\end{align*}
where the first inequality uses the monotonicity of Rényi divergences \cite[Theorem 7]{muller-lennert_quantum_2013}, and $\divgarg{\alpha}{\cdot}{\cdot}$ denotes the quantum Rényi divergence with order $\alpha$ which encompasses the quantum relative entropy as the special case $\alpha=1$. Here, $\divgarg{2}{\rho}{\sigma} := \ln \tr\big(\sigma^{-1/2} \rho \sigma^{-1/2} \rho\big)$. We also used the fact that the reference state $\pi_N$ is the maximally mixed state in the second line above to make simplifications. The second last inequality follows from the eigenvalue decomposition $\ce(\rho) - \pi_N = \sum_{k} \lambda_k \ket{\phi_k}\bra{\phi_k}$, so that $\tr\big((\ce(\rho) - \pi_N)^2\big) = \sum_{k} \lambda_k^2 \le \big(\sum_k \abs{\lambda_k}\big)^2 = \big(\tr\abs{\ce(\rho)-\pi_N}\big)^2$. The last line uses the quantum Pinsker inequality (see e.g., \cite[Theorem 11.9.1]{wilde_quantum_2017}).

Consequently,
\begin{align*}
\frac{\divgklbig{(1-x) \pi_N + x \ce(\rho)}{\pi_N}}{d(x||\eps) + x \divgkl{\rho}{\pi_N}} \le 2\ln(2)\ x N \frac{\divgkl{\ce(\rho)}{\pi_N}}{\divgkl{\rho}{\pi_N}} \le 2\ln(2)\ x N \le \frac{2\ln(2)\ N}{\ln(1/\eps)^2},
\end{align*}
where the second inequality follows from the data processing inequality and the third from the choice $x \le x^\star$. The above argument is valid when $\rho \neq \pi_N$. The remaining case $x \neq \eps, \rho = \pi_N$ can be validated directly, since the numerator of the ratio is simply zero. In summary, we have
\begin{align}
\label{eqn::case1}
\frac{\divgklbig{(1-x) \pi_N + x \ce(\rho)}{\pi_N}}{d(x||\eps) + x \divgkl{\rho}{\pi_N}} \le \frac{2\ln(2)\ N}{\ln(1/\eps)^2}.
\end{align}

\emph{Case 2:} consider $x > {x}^{\star}$. By the convexity of quantum relative entropy \cite[Corollary 11.9.2]{wilde_quantum_2017},
\begin{align}
\label{eqn::case2_1}
\begin{aligned}
&\ \divgklbig{(1-x) \pi_N + x \ce(\rho)}{\pi_N} \\
\le &\  (1-x) \divgkl{\pi_N}{\pi_N} + x \divgklbig{\ce(\rho)}{\pi_N} \\
=&\ x \divgklbig{\ce(\rho)}{\pi_N}.
\end{aligned}
\end{align}
We also have
\begin{align*}
d(x||\eps) &=  x\ln (x/\eps) + (1-x) \ln \big((1-x)/(1-\eps)\big) \\
& = x\big(\ln(x) - \ln(\eps)\big) + (1-x) \ln(1-x) - (1-x)\ln(1-\eps) \\
& \ge x\big(\ln(x) - \ln(\eps)\big) + (1-x) \ln(1-x) \\
& \ge  x\big(\ln(x) - \ln(\eps)\big) - x \\
& = x \big( \ln(x) - 1 - \ln(\eps)\big).
\end{align*}
The inequality $(1-x) \ln(1-x) \ge -x$ can be readily verified for $x\in [0,1]$.
By the choice $x > x^\star$, one has
\begin{align}
\label{eqn::case2_2}
d(x||\eps) \ge x \Big(\ln (\frac{1}{(\ln(1/\eps))^2}) - 1 - \ln(\eps)\Big).
\end{align}
This lower bound is positive for small enough $\eps$.
Combining the last two estimates \eqref{eqn::case2_1} and \eqref{eqn::case2_2}, one has
\begin{align*}
\frac{\divgklbig{(1-x) \pi_N + x \ce(\rho)}{\pi_N}}{d(x||\eps) + x \divgkl{\rho}{\pi_N}} & \le \frac{x \divgkl{\ce(\rho)}{\pi_N}}{d(x||\eps)} \\
& \le \frac{\divgkl{\ce(\rho)}{\pi_N}}{\ln (\frac{1}{(\ln(1/\eps))^2}) - 1 + \ln(1/\eps)}.
\end{align*}
By taking the sup over all $\rho$, one has
\begin{align}
\label{eqn::case2}
\begin{aligned}
\sup_{\rho} \frac{\divgklbig{(1-x) \pi_N + x \ce(\rho)}{\pi_N}}{d(x||\eps) + x \divgkl{\rho}{\pi_N}} \le &\  \frac{ \sup_{\rho} \divgkl{\ce(\rho)}{\pi_N}}{\ln (\frac{1}{(\ln(1/\eps))^2}) - 1 + \ln(1/\eps)} \\
=&\ \frac{C}{\ln (\frac{1}{(\ln(1/\eps))^2}) - 1 + \ln(1/\eps)}.
\end{aligned}
\end{align}

Finally, by summarizing the above two cases, we obtain
\begin{align*}
\eta_{\kl}({\cn}, \widehat{\sigma}_\eps) \myle{\eqref{eqn::case1},\eqref{eqn::case2}} & \max\big\{\ \frac{2\ln(2)\ N}{\ln(1/\eps)^2}, \frac{C}{\ln (\frac{1}{(\ln(1/\eps))^2}) - 1 + \ln(1/\eps)}\ \big\}.
\end{align*}
For small enough $\eps$, $\ln(1/\eps)$ is the leading term, and therefore,
\begin{align*}
\limsup_{\eps\to 0}\ \eta_{\kl}({\cn}, \widehat{\sigma}_\eps) \ln(1/\eps) \le  C.
\end{align*}

\section{Conclusion and outlook}

In this work, we report two negative results, showing that the tensorization for SDPI fails for certain quantum chi-square divergences and for quantum relative entropy. 
These results suggest that the symmetric quantum chi-square divergence (with a proven tensorization property) remains a valuable mathematical tool in general.
Several questions remain open for future investigation:

\begin{itemize}

\item \emph{(Characterizing the transition).} By the same argument as in Corollary~\ref{corollary}, within the mean-alpha family for $\alpha \in (0,1)$, the set of $\alpha$ such that a counterexample exists is an open set. It remains unknown how the transition occurs between the region near $\alpha = 0$ (where tensorization does not hold), and the special case $\alpha = 1/2$ (where tensorization does hold). \smallskip

\item \emph{(Restriction to special channels).} In some settings, one does not need to study all possible channels, and only a particular family of quantum channels is of interest. In that case, one may still have the tensorization of SDPI and it would be interesting to establish this.
\end{itemize}

\section*{Acknowledgement}

During the investigation, the large language model was used for exploration, and initially generated the counterexamples and proofs for Propositions~\ref{prop::chisq} and \ref{prop::kl}.
The AI-generated content was independently verified by the author.
The author prepared the manuscript presenting these examples and proofs.
This work was partially supported by Shanghai Pilot Program for Basic Research and Shanghai Jiao Tong University 2030 Initiative.

\bibliographystyle{abbrvnat}
\bibliography{ref.bib}

\end{document}